\documentclass[conference,twocolumn,10pt]{IEEEtran} 
\usepackage{epsfig,cite}
\usepackage{graphicx}
\usepackage{color}
\usepackage{amssymb,amsmath,mathrsfs}
\definecolor{blue}{rgb}{0,0,1}
\definecolor{darkgreen}{rgb}{0,.5,0}
\definecolor{darkred}{rgb}{.5,0,0}

\IEEEoverridecommandlockouts

\newtheorem{theorem}{Theorem}

\newtheorem{ass}{Assumption}

\newcommand{\given}{\:|\:}

\def\Pr{{\mathrm{Pr}}}

\begin{document}

\title{Thermodynamic Properties of\\ Molecular Communication}

\author{
	\IEEEauthorblockN{Andrew W. Eckford$^1$, Benjamin Kuznets-Speck$^2$, Michael Hinczewski$^3$, and Peter J. Thomas$^4$}
	\IEEEauthorblockA{$^1$Dept. of EECS, York University, Toronto, ON, Canada\\
    $^{2,3}$ Dept. of Physics, Case Western Reserve University, Cleveland, OH, USA\\
    $^4$ Dept. of Math., Appl. Math., and Stats., Case Western Reserve University, Cleveland, OH, USA\\
    Emails: $^1$aeckford@yorku.ca, $^2$bdk48@case.edu, $^3$mxh605@case.edu $^4$pjthomas@case.edu}
	\thanks{This work was supported by a grant from the Natural Sciences and Engineering Research Council, and  by NSF grants DMS-1413770, DEB-1654989, DMS-1440386, and MCB-1651560.}}

\maketitle

\begin{abstract}
In this paper, we consider the energy cost of communicating using molecular communication. In a simplified scenario, we show that the energy bound in Landauer's principle can be achieved, implying that molecular communication can approach fundamental thermodynamic limits. 
\end{abstract}

\section{Introduction}

In molecular communication, chemical principles are used to communicate: for example, a message may be carried by molecules diffusing in a medium, propagating via Brownian motion to a receiver \cite{pierobon10,Farsad16}.
Various forms of molecular communication are found in nature, such as signal transduction;
these mechanisms are widely used and are central to the function of many biological systems.

Since microorganisms expend significant energy to transmit and process information, information-theoretic analysis of molecular communication \cite{Farsad17,Rose16,Thomas16,Sakkaff17} may be used to analyze the efficiency of biological processes.
Microscopic information processing systems have also been considered from a thermodynamic perspective. Mutual information and thermodynamic entropy production have been shown to be related for information transfer via Brownian motion \cite{Allahverdyan09}. This concept has been generalized to the ``learning rate'' (time derivative of conditional Shannon entropy), which is related to energy consumption in bipartite signaling networks \cite{Barato14}. In a different setting, a thermodynamic analysis of molecular communication was performed in \cite{Pierobon13}.


In this paper we are motivated by Landauer's principle \cite{Landauer61}, which gives the minimum free energy cost in order to erase one bit of information.  Thus any information processing system that involves erasure as part of its mechanism will need a free energy input in order to operate.  Landauer's principle has been analyzed in computation \cite{Bennett03} and biological information processing \cite{Kempes17}, and has been validated in experiments \cite{Berut12}.

To analyze this phenomenon in molecular communication, we consider a ``minimal'' molecular communication system: two reservoirs are held at different concentrations, and binary transmissions are composed by selecting molecules from the reservoirs. The free energy stored in the reservoirs is expressed in terms of the {\em chemical potential} between reservoirs at different concentrations. Our main result (Theorem 1) shows that the Landauer principle gives an achievable lower bound on the free energy required to create our molecular communication transmitter; equivalently, we show that the reciprocal of the Landauer energy is an achievable upper bound on the capacity per unit energy cost (cf. \cite{verdu90}).

\section{Physical Model and Concepts}

\subsection{Chemical potential}

Throughout this paper, we express mixtures of solvent and solute in terms of mole fraction $c$. That is, $c$ is the ratio of the number of solute molecules to the total number of molecules in the solution. 

Consider a closed container, divided into two reservoirs. The reservoirs contain solutions of a given solute at different mole fractions: the two mole fractions are a {\em low} mole fraction $c_L$ and a {\em high} mole fraction $c_H$, where $0<c_L \leq c_H$. We refer to these as the low and high reservoirs, respectively. 

The {\em chemical potential} difference between the reservoirs gives the free energy required to move a single solute molecule from one to the other. 
We assume the reservoirs contain {\em ideal solutions}, mixtures of solute and solvent where the interaction energies between neighboring particles are independent of particle type.  Under this assumption the chemical potential difference from the low to the high reservoir is a function of the solute mole fractions \cite{AtkinsPhysChem},
%
%
\begin{align}
	\mu&= 
	\label{eqn:PerMoleculePotential}
    kT \log \frac{c_H}{c_L} ,
\end{align}
where $k$ is Boltzmann's constant and $T$ the absolute temperature.  When $c_L < c_H$ the potential $\mu$ going from low to high is positive.  The reverse potential (from high to low) is $- \mu$; moving a solute particle from high to low decreases free energy, while moving from low to high increases free energy.


\subsection{Landauer's principle}

The Maxwell's Demon thought experiment \cite{LeffRexBook} postulates a being with the ability to sense and manipulate individual atoms in a gas. Measuring the velocity of each atom, the being opens a trapdoor at just the right times to capture ``hot" atoms (moving faster than average) on one side, and ``cold" atoms (moving slower than average) on the other (Fig. \ref{fig:MaxwellsDemon}). If the trapdoor were designed so as to require no net energy to operate, then the thermodynamic entropy of the system could be reduced with no energy cost; this would be an apparent violation of the second law of thermodynamics.

\begin{figure}
\begin{center}
\includegraphics[width=3in]{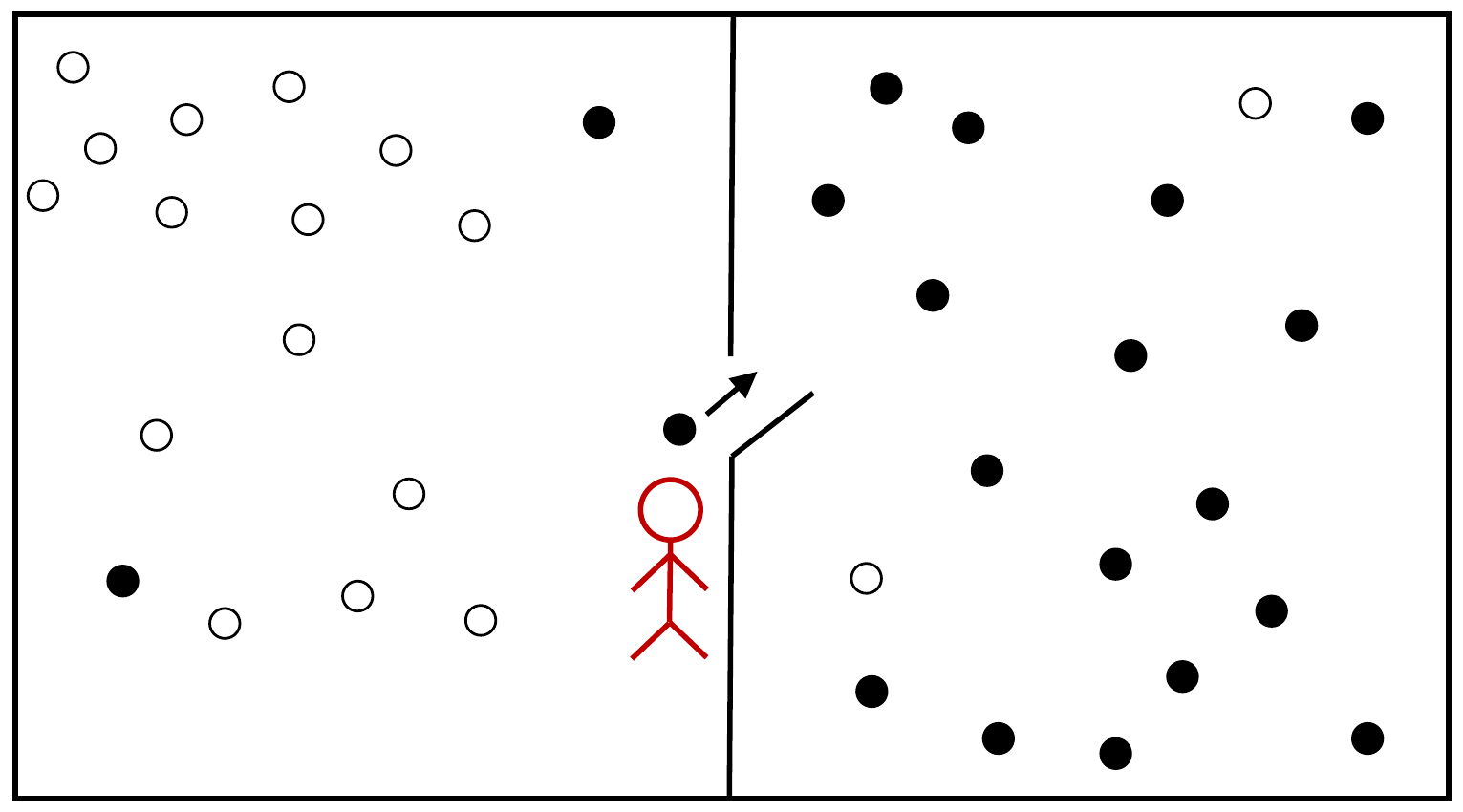}
\end{center}
\caption{\label{fig:MaxwellsDemon} A depiction of the Maxwell's Demon thought experiment. The demon (stick figure) can monitor the atoms in a gas (filled circles are moving faster than average, unfilled are moving slower than average), and can operate a trapdoor. If it detects a fast-moving atom near the trapdoor moving to the right, it opens the trapdoor, allowing the atom to pass through the barrier (and similarly for slow-moving atoms going left).  Thus, the temperature on the right of the partition increases, and on the left decreases, seemingly violating the second law of thermodynamics.}
\end{figure}

A more careful analysis, however, rescues the second law by considering the thermodynamics of the demon and the physical nature of information itself \cite{Parrondo15}.  In a landmark work, Landauer \cite{Landauer61} showed that while the demon could in principle carry out the necessary measurements (i.e. analyzing particle trajectories) without costing free energy, the resulting measurement information must be stored in a physical memory register for some finite amount of time. If we imagine the demon operating indefinitely, this information must eventually be erased, since any physical memory register has a limited storage capacity.  The {\em erasure} of one {\em bit} of information costs at least
\begin{equation}
	G_\ell = k T \log 2
\end{equation}
Joules of free energy, where $\log$ is the natural logarithm.

If information is measured in {\em nats}, Landauer's principle becomes
\begin{equation}
	\label{eqn:LandauerEnergy}
	G_\ell = k T
\end{equation}
joules per nat (since there are $\log 2$ nats per bit). We refer to $kT$ as the {\em Landauer energy}.

\section{Communication model}

\subsection{Minimal Molecular Communication Channel}

%
%
Using the low and high reservoirs described in the previous section,
suppose the total numbers of molecules in the low and high reservoirs are $n_L$ and $n_H$, respectively.
The average mole fraction across the entire container (counting low and high reservoirs together) is
\begin{equation}
	c = \frac{n_L c_L + n_H c_H}{n_L + n_H} .
\end{equation}
Moreover, let $n = n_L+n_H$ represent the total number of molecules.

Consider the following simple molecular communication system, depicted in Fig. \ref{fig:CommunicationSystem}:
\begin{itemize}
	\item {\em Transmitter:} The transmitter sends a single bit $X \in \{0,1\}$ as follows: if $X = 1$, the selector picks a molecule from the high reservoir; if $X = 0$, the selector picks a molecule from the low reservoir. {\em The molecule might be either solvent or solute.}
	\item {\em Receiver:} The receiver detects whether the selected molecule is solvent or solute, and forms $Y \in \{0,1\}$: $Y = 0$ if it is solvent, and $Y = 1$ if it is solute. After detection, the receiver consumes the molecule. The molecules are selected independently of their identity (solute or solvent).
    \item The total number of molecules $n$ (solute + solvent) decreases by one with every use of the channel; we neglect this effect since $n$ is very large, so $c_L$ and $c_H$ remain the same after each channel use.
\end{itemize}
We refer to this as the {\em minimal molecular communication} (MMC) channel.

\begin{figure}
\begin{center}
\includegraphics[width=3.5in]{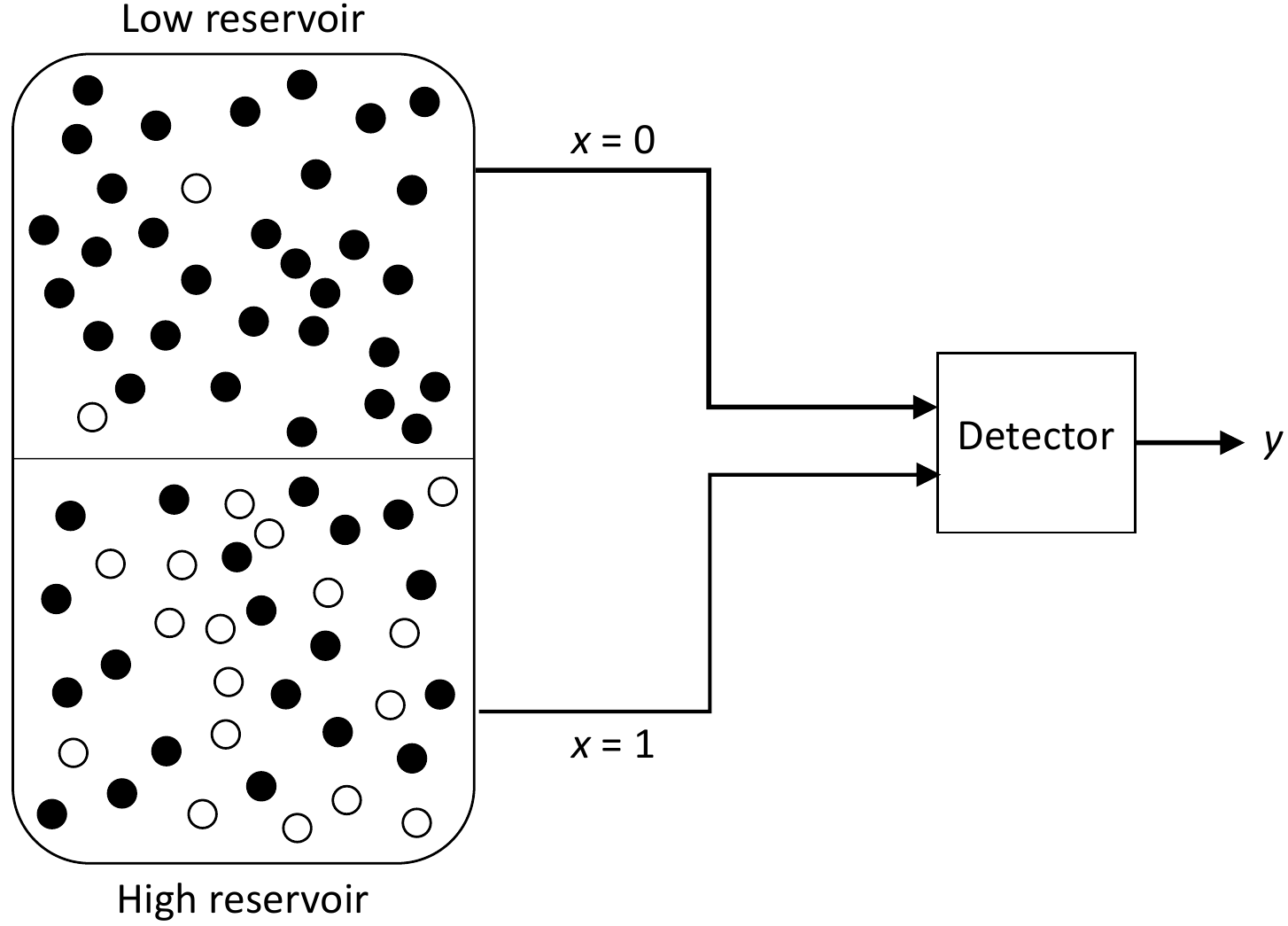}
\end{center}
\caption{\label{fig:CommunicationSystem} A depiction of the communication system. Solvent molecules are filled circles, while solute molecules are unfilled circles. The transmitter selects a single molecule from the low ($x=0$) or high ($x=1$) reservoir, and provides it to the detector; the detector decides $y = 0$ if the molecule is solvent, or $y = 1$ if the molecule is solute. For ease of illustration, Assumption \ref{ass:SmallC} does not apply to this figure.}
\end{figure}

A {\em single channel use} consists of the selection, and detection, of a single molecule.
Since we are using mole fractions, we have:
\begin{equation}
	p(y \given x = 0) = 
	\left\{
		\begin{array}{cl}
			1-c_L, & y = 0 \\
			c_L, & y = 1
		\end{array}
	\right.
\end{equation}
and
\begin{equation}
	p(y \given x = 1) = 
	\left\{
		\begin{array}{cl}
			1-c_H, & y = 0 \\
			c_H, & y = 1 .
		\end{array}
	\right.
\end{equation}
Throughout this paper we make the following assumption.
\begin{ass}
	\label{ass:SmallC}
	Solute mole fractions are small, i.e., $c_L,c_H \ll 1$.
\end{ass}

\subsection{Information rate per channel use}

In each channel use, the channel is a memoryless (but asymmetric) binary channel (Fig. \ref{fig:BinaryChannel}). 
Let $p_L = \Pr(X = 0)$, $(1-p_L) = \Pr(X=1)$, and 
\begin{equation}
	w = p_L c_L + (1-p_L) c_H .
\end{equation} 
Note that $w$ (the average mole fraction in a {\em codeword}) is distinct from $c$ (the
latter is the average mole fraction in the {\em container}).
Then
\begin{align}
	p(y) &= p_L p(y \given x = 0) + (1-p_L) p(y \given x = 1) \\
	&=\left\{
		\begin{array}{cl}
			1-w, & y = 0 \\
			w, & y = 1
		\end{array}
	\right.
\end{align}
%
%
Finally
\begin{align}
	\label{eqn:MutualInfo}
	I(X;Y) &= \mathscr{H}(w) - p_L \mathscr{H}(c_L) - (1-p_L) \mathscr{H}(c_H) ,
\end{align}
where $\mathscr{H}(\cdot)$ is the binary entropy function (using the natural logarithm in this case). 

\begin{figure}
\begin{center}
\includegraphics[width=2.5in]{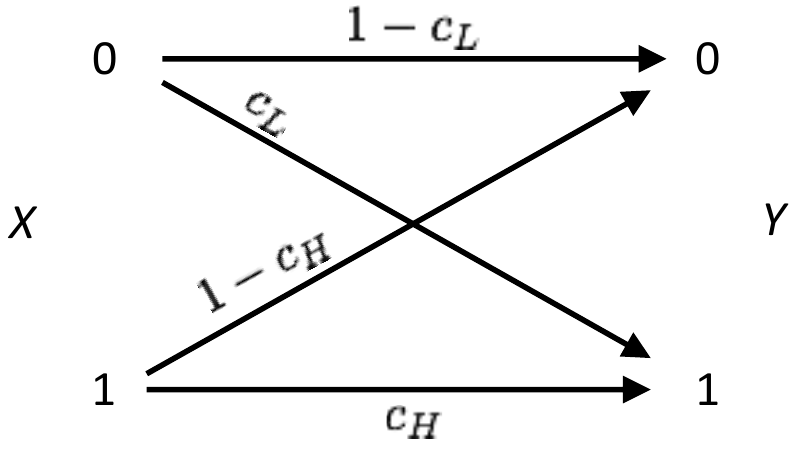}
\end{center}
\caption{\label{fig:BinaryChannel} A depiction of the MMC channel as a binary asymmetric channel.}
\end{figure}

The binary entropy function can be written
\begin{equation}
	\mathscr{H}(p) = p - p \log p - o(p^2) ,
\end{equation}
so $\mathscr{H}(p) = p - p \log p$ is asymptotically accurate when $p$ is small.
Using Assumption \ref{ass:SmallC}, the {\em total} information that can be communicated in $n$ channel uses is
\begin{equation}
	\label{eqn:MutualInfoSmall}
	I = n\Big(p_L c_L \log c_L + (1-p_L) c_H \log c_H - w \log w \Big) .
\end{equation}
%




\section{Energy per unit information}








\subsection{Free energy required to create the reservoirs}

Suppose we start in a solution at initial mole fraction $c$, and we fill both the low and high reservoirs from this solution (so that, initially, $c_L = c_H = c$). 
To implement the communication system we just described, we want to move solute molecules from low to high reservoir, until the final mole fractions are $c_L < c$ and $c_H > c$. The chemical potential (\ref{eqn:PerMoleculePotential}) implies that this operation requires free energy; thus, we can determine how much free energy it takes to create our communication system, and ultimately how many joules per nat are required to create the transmitter.

Starting at initial mole fraction $c$, the total number of solute molecules moved $m$ is 
\begin{equation}
	\label{eqn:M}
	m = n_L(c - c_L) = n_H(c_H - c) .
\end{equation}
Supposing we move $\Delta m$ molecules at a time, this operation requires $m/\Delta m$ moves.  In the derivation below, we choose $\Delta m \ll n_H, n_L$, which allows us to set up a Riemann integral.  By splitting up the operation into many small moves, we also implicitly assume that the whole process happens slowly enough that the two reservoirs each remain instantaneously at equilbrium at all times.  The resulting free energy that we derive is thus the smallest possible amount of work needed to create the transmitter.  Any faster, non-equilibrium alternative procedure will necessarily require more work. 

%
%

%
%

Assumption \ref{ass:SmallC} implies that $n_L$ and $n_H$ are much larger than the total number of solute molecules in each reservoir (since $n_L$ and $n_H$ include both solute and solvent). From this assumption, it follows that, e.g., $n_L - m \simeq n_L$, so that the moves of solute molecules do not significantly change $n_L$ (nor, by a similar argument, $n_H$); thus, we treat $n_L$ and $n_H$ as constant.

Initially both reservoirs are at mole fraction $c$, so $\mu = 0$ (from (\ref{eqn:PerMoleculePotential})); i.e., no energy is required for the first move. However, once the first move is complete, the mole fractions of the high and low reservoirs are
\begin{align}
	c_{H,1} &= c + \frac{\Delta m}{n_H} \\
	c_{L,1} &= c - \frac{\Delta m}{n_L} .
\end{align}
%
%
%
%
After the first move, the chemical potential (per molecule) is (from (\ref{eqn:PerMoleculePotential}))
\begin{align}
	\mu_1 
	&= k T \log \frac{c_{H,1}}{c_{L,1}} \\
	\label{eqn:E1-initial}
	&= kT \log \frac{c + \Delta m / n}{c - \Delta m / n} 
\end{align}
so that the free energy required to move an additional $\Delta m$ molecules is
\begin{equation}
	G_1 = \Delta m \: kT \log \frac{c + \Delta m / n_H}{c - \Delta m / n_L} .
\end{equation}
It follows that the energy required for the $j$th move is
\begin{equation}
	G_j = \Delta m \: kT \log \frac{c + j \Delta m / n_H}{c - j \Delta m / n_L} ,
\end{equation}
and the total energy required over all $m/\Delta m$ moves is 
\begin{align}
	G &= \sum_{j=0}^{\frac{m}{\Delta m} - 1} G_j \\
	&= \sum_{j=0}^{\frac{m}{\Delta m}-1} \Delta m \: kT \log \frac{c + j \Delta m / n_H}{c - j \Delta m / n_L} \\
	&= n_H kT \sum_{j=0}^{\frac{m}{\Delta m}-1}\log \frac{c + j \Delta m / n_H}{c - j \Delta m / n_L}\:\frac{\Delta m}{n_H}
\end{align}
Substituting $\Delta c = \frac{\Delta m}{n_H}$ as the change in mole fraction in the high reservoir,
\begin{equation}
	\label{eqn:EnergySum}
	G = n_H kT \sum_{j=0}^{\frac{m}{\Delta m}-1}\log \frac{c + j \Delta c}{c - \frac{n_H}{n_L} j \Delta c}\:\Delta c .
\end{equation}
As $\Delta c \rightarrow 0$, (\ref{eqn:EnergySum}) becomes a Riemann integral. Replacing $j \Delta c$ with
$\chi$ as the variable of integration,
\begin{align}
	\nonumber \lefteqn{\frac{G}{n_H kT} = \int_{\chi = 0}^{c_H-c} \log \frac{c + \chi}{c - \frac{n_H}{n_L}\chi} d\chi}&\\
	\nonumber
	&= \left((c + \chi)\log(c+\chi) - \chi\right) \Big|_{\chi = 0}^{c_H - c} \\
	&\:\:\:\: - 
	\left(- \left(\frac{n_L}{n_H}c - \chi\right)\log \left(c - \frac{n_H}{n_L}\chi\right) - \chi \right) \Big|_{\chi = 0}^{c_H-c} \\
	\label{eqn:E2}
	&= - \left(1 + \frac{n_L}{n_H}\right) c \log c + c_H \log c_H + \frac{n_L}{n_H} c_L \log c_L .
\end{align}
%
%
To normalize the expression in (\ref{eqn:E2}), we divide it by $1 + \frac{n_L}{n_H}$ and obtain
%
\begin{align}
	\nonumber \lefteqn{\frac{G}{n_HkT(1+\frac{n_L}{n_H})} = \frac{G}{kT(n_H + n_L)}}&\\ 
	\label{eqn:Energy-Different}
	&= \frac{n_H}{n_H + n_L} c_H \log c_H + \frac{n_L}{n_H + n_L} c_L \log c_L - c \log c .
\end{align}
Since $n_H+n_L = n$, the total number of molecules in the container,
$\frac{G}{n_H+n_L}$ (easily obtained from (\ref{eqn:Energy-Different})) gives the free energy per molecule to establish the communication system.

\subsection{Capacity per unit energy}

To make the notation more compact, let
$m_L = \frac{n_L}{n_H+n_L}$
represent the fraction of total molecules in the low reservoir, and let 
\begin{equation}
	\phi(p) = \left\{ \begin{array}{cl} 0, & p = 0\\ p \log p, & p > 0 \end{array} \right. 
\end{equation}
represent the partial entropy function.
Recalling $n = n_L + n_H$, (\ref{eqn:Energy-Different}) can be written
\begin{equation}
	\label{eqn:Energy-Different-m0}
	G = n kT \Big( m_L \phi(c_L) + (1-m_L) \phi(c_H) - \phi(c) \Big).
\end{equation}
%
%
%
If $m_L = p_L$, we can write
\begin{equation}
    \label{eqn:fraction}
	\frac{G}{I} =
	kT \frac{m_L \phi(c_L) + (1-m_L) \phi(c_H) - \phi(c)}
		{p_L \phi(c_L) + (1-p_L) \phi(c_H) - \phi(w)},
\end{equation}
and since $c = w$ when $m_L = p_L$,
\begin{equation}
	\label{eqn:kT}
    \frac{G}{I} = kT
\end{equation}
joules per nat.
That is, {\em the energy required to create the transmitter is equal to the Landauer energy}.

However, (\ref{eqn:fraction}) does not apply if $m_L \neq p_L$: 
for a sufficiently long codeword, one reservoir will run out before the other; the remaining molecules in the other reservoir can't be used to send information.
That is, the reservoir is {\em mismatched to the codebook}.

\begin{figure}
\begin{center}
\includegraphics[width=3.5in]{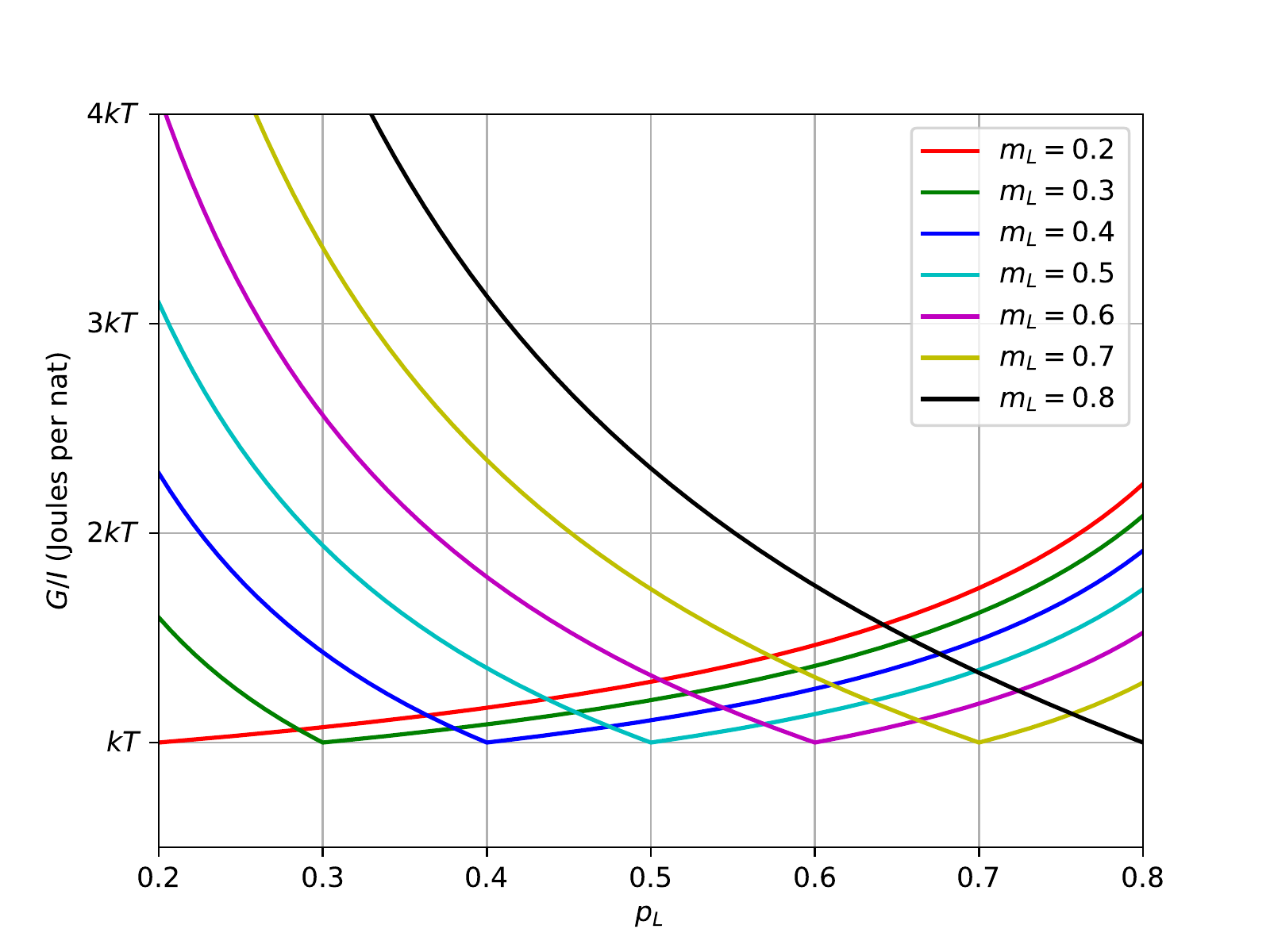}
\end{center}
\caption{\label{fig:EnergyCapacityResult} Plot of $G/I$ versus $p_L$ for various values of $m_L$, showing that the minimum $G/I = kT$ occurs at $m_L = p_L$. Where $m_L \neq p_L$, $G/I$ is given by (\ref{eqn:AdjustedCE1}) and (\ref{eqn:AdjustedCE2}).}
\end{figure}

Nonetheless, the Landauer energy is a lower bound on $G/I$, as we show 
in the main result in this paper:
\begin{theorem}
	For the MMC system, the minimum energy per nat is $kT$, achieved when $m_L = p_L$.
\end{theorem}
\begin{proof}
Achieving $kT$ at $m_L = p_L$ is clear from the development up to (\ref{eqn:kT}), so we focus on $m_L \neq p_L$.

First consider $m_L < p_L$, where the low reservoir runs out before the high reservoir. Let $n_H^\prime$ represent the number of molecules used from the high reservoir. Since $p_L = n_L / (n_L + n_H^\prime)$, then $n_H^\prime = n_L (1-p_L)/p_L$. As a function of $n_L$ and $p_L$, the total number of molecules {\em used} is $n_L + n_L (1-p_L)/p_L = n_L/p_L$; once this number is used, the low reservoir has run out (and further communication is impossible: we are no longer able to send the symbol 0, corresponding to the low reservoir).

When the reservoirs are created, the required energy is given by (\ref{eqn:Energy-Different-m0}). This can be rewritten
\begin{align}
	G &= \frac{n_L}{m_L} kT \Big( m_L \phi(c_L) + (1-m_L) \phi(c_H) - \phi(c) \Big) .
\end{align}
However when the reservoirs are used to communicate, if $m_L < p_L$ we can obtain
\begin{equation}
	I = \frac{n_L}{p_L} \Big( p_L\phi(c_L) + (1-p_L) \phi(c_H) - \phi(w) \Big) .
\end{equation}
Thus,
\begin{equation}
	\label{eqn:AdjustedCE1}
	\frac{G}{I} = kT \left(\frac{p_L}{m_L}\right) \frac{m_L \phi(c_H) + (1-m_L) \phi(c_L) - \phi(c)}
		{p_L \phi(c_L) + (1-p_L) \phi(c_H) - \phi(w)} .
\end{equation}

Similarly, if $m_L > p_L$, then the high reservoir runs out first, and the total number of molecules used is $n_H/(1-p_L)$.
Now we have
\begin{equation}
	\label{eqn:AdjustedCE2}
	\frac{G}{I} = kT \left(\frac{1-p_L}{1-m_L}\right) \frac{m_L \phi(c_H) + (1-m_L) \phi(c_L) - \phi(c)}
		{p_L \phi(c_L) + (1-p_L) \phi(c_H) - \phi(w)} .
\end{equation}

In the appendix, we show that
\begin{equation}
	\frac{p_L \phi(c_L) + (1-p_L) \phi(c_H) - \phi(w)}{p_L}
\end{equation}
is a monotonically increasing function in $p_L$, and
\begin{equation}
	\frac{p_L \phi(c_L) + (1-p_L) \phi(c_H) - \phi(w)}{1-p_L}
\end{equation}
is a monotonically decreasing function in $p_L$. Thus, in both (\ref{eqn:AdjustedCE1}) and (\ref{eqn:AdjustedCE2}), 
$\frac{G}{I} > kT$ for $p_L \neq m_L$, and the theorem follows.
\end{proof}


In Fig. \ref{fig:EnergyCapacityResult}, we illustrate $G/I$ using (\ref{eqn:AdjustedCE1}) and (\ref{eqn:AdjustedCE2}).

\section{Discussion and Conclusion}

In a simplified molecular communication system, this paper found that the minimum energy per unit information is given by the Landauer energy of $kT$ joules per nat. Thus, molecular communication is, in principle, an energy-efficient form of communication; this may partly explain why it is used in biological systems. 

It should be noted that many simplifying assumptions were used in developing this result, such as the assumption that molecules are transmitted perfectly to the receiver. If diffusion were used instead, this may have a deleterious effect on the system, potentially driving up the energy requirements. This will be considered in future work.

%

\appendix

%
%


For compact notation, let
\begin{align}
	\nonumber\lefteqn{J(c_L,c_H,p_L)}&\\
	& = p_L\phi(c_L) + (1-p_L)\phi(c_H) - \phi(p_L c_L + (1-p_L)c_H).
\end{align}
Here we show that
\begin{equation}
	\label{eqn:Monotonic-1}
	\frac{1}{p_L} J(c_L,c_H,p_L)
\end{equation}
is a monotonically decreasing function in $p_L$, and that
\begin{equation}
	\label{eqn:Monotonic-1a}
	\frac{1}{1-p_L} J(c_L,c_H,p_L)
\end{equation}
is a monotonically increasing function in $p_L$.

Starting with (\ref{eqn:Monotonic-1}), taking the first derivative, we have
\begin{align}
	\frac{d}{dp_L} \frac{J(c_L,c_H,p_L)}{p_L}  
	\label{eqn:Monotonic-2}
	&=\frac{p_L \frac{d}{dp_L}J(c_L,c_H,p_L) - J(c_L,c_H,p_L)}{p_L^2} 
\end{align}
where
\begin{align}
    \frac{d}{dp_L}J(c_L,c_H,p_L)
	&= \phi(c_L) - \phi(c_H) + (\log(w) + 1)(c_H - c_L) .
\end{align}
We are only interested in whether (\ref{eqn:Monotonic-2}) is positive or negative, which is determined by its numerator; the denominator is 
always positive and can be ignored. The numerator becomes
\begin{align}
	\nonumber \lefteqn{
    p_L \frac{d}{dp_L} J(c_L,c_H,p_L) - J(c_L,c_H,p_L)}&\\
	\label{eqn:Monotonic-3}
	&= - c_H \sum_{i=2}^\infty \frac{p_L^i \left(1-\frac{c_L}{c_H}\right)^i}{i}.
\end{align}
which follows from the Taylor series expansion of $\log(1-x)$.
Since $c_L < c_H$ by definition, the term in (\ref{eqn:Monotonic-3}) is always negative, and is always within the range of convergence of the Taylor series. Thus, (\ref{eqn:Monotonic-1}) is decreasing in $p_L$.

Returning to (\ref{eqn:Monotonic-1a}), we have
\begin{align}
	\nonumber \lefteqn{\frac{d}{dp_L} \frac{1}{1-p_L} J(c_L,c_H,p_L) } &\\
	\label{eqn:Monotonic-4}
	&=\frac{(1-p_L) \frac{d}{dp_L}J(c_L,c_H,p_L) + J(c_L,c_H,p_L)}{(1-p_L)^2}. 
\end{align}
The derivation is similar. The 
numerator in (\ref{eqn:Monotonic-4}) becomes
\begin{align}
	c_L(1-p_L)\left(\frac{c_H}{c_L}-1\right)
    - c_L\log \left( 1 + (1-p_L)\left(\frac{c_H}{c_L}-1\right)\right)
	\label{eqn:Monotonic-5}
\end{align}
which (using Taylor series) can be expanded to
\begin{align}
	\label{eqn:Monotonic-6}
	c_L \sum_{i=2}^\infty \frac{(1-p_L)^i \left(\frac{c_H}{c_L} - 1\right)^i}{i} (-1)^i ,
\end{align}
where (\ref{eqn:Monotonic-6}) is valid for $(1-p_L)\left(\frac{c_H}{c_L} - 1\right) < 1$. 
Using the Leibniz criterion for alternating series, it can be shown that 
(\ref{eqn:Monotonic-6}) is always positive. For $(1-p_L)\left(\frac{c_H}{c_L} - 1\right) \geq 1$, return to (\ref{eqn:Monotonic-5}), let $x = (1-p_L)\left(\frac{c_H}{c_L} - 1\right)$ and see that $c_L(x - \log(1+x)) > 0$ for all $x \geq 1$. Thus, (\ref{eqn:Monotonic-5}) is positive, and (\ref{eqn:Monotonic-1a}) is increasing in $p_L$.


\bibliographystyle{ieeetr}
\bibliography{ThermoMolCom}

\end{document}